\documentclass[11pt,aps,pre,a4paper,byrevtex,showpacs,showkeys,longbibliography,notitlepage]{revtex4-1}
\usepackage{microtype}
\usepackage{amsmath}
\usepackage{amssymb}
\usepackage{amsthm}
\usepackage{color}
\definecolor{myblue}{rgb}{0.153,0.322,0.706}
\usepackage[colorlinks,linkcolor=myblue,urlcolor=myblue,citecolor=myblue,breaklinks=true]{hyperref}

\newcommand{\cE}{\mathcal{E}}
\newcommand{\cM}{\mathcal{M}}

\newcommand{\cY}{\mathcal{Y}}
\newcommand{\cN}{\mathcal{N}}
\newcommand{\barm}{\bar m}
\newcommand{\eps}{\epsilon}
\newcommand{\cin}{\,\cdot\,}
\newcommand{\idf}{\mathbf{1}}
\newcommand{\ra}{\rightarrow}
\newcommand{\be}{\begin{equation}}
\newcommand{\ee}{\end{equation}}
\newcommand{\om}{\omega}
\newcommand{\reals}{\mathbb{R}}

\newtheoremstyle{myplain}
{5pt}			%Space above 
{5pt}			%Space below 
{\normalsize}	%Body font 
{}			%Indent amount (empty = no indent, \parindent = para indent) 
{\bfseries}		%Thm head font 
{:}			%Punctuation after thm head 
{.5em}		%Space after thm head: " " = normal interword space; \newline = linebreak 
{\thmname{#1}\thmnumber{ #2}\thmnote{~{(#3)}}}

\theoremstyle{plain}
\newtheorem{thm}{Theorem}
\newtheorem{prop}[thm]{Proposition}

\newtheorem{hypothesis}{Hypotheses}

\theoremstyle{myplain}
\newtheorem*{remark*}{Remarks}

\begin{document}
\title{Asymptotic equivalence of probability measures and stochastic processes}

\author{Hugo Touchette}
\affiliation{National Institute for Theoretical Physics (NITheP), Stellenbosch 7600, South Africa}
\affiliation{Institute of Theoretical Physics, Department of Physics, University of Stellenbosch, Stellenbosch 7600, South Africa}

\date{\today}

\begin{abstract}
Let $P_n$ and $Q_n$ be two probability measures representing two different probabilistic models of some system (e.g., an $n$-particle equilibrium system, a set of random graphs with $n$ vertices, or a stochastic process evolving over a time $n$) and let $M_n$ be a random variable representing a ``macrostate'' or ``global observable'' of that system. We provide sufficient conditions, based on the Radon--Nikodym derivative of $P_n$ and $Q_n$, for the set of typical values of $M_n$ obtained relative to $P_n$ to be the same as the set of typical values obtained relative to $Q_n$ in the limit $n\rightarrow\infty$. This extends to general probability measures and stochastic processes the well-known thermodynamic-limit equivalence of the microcanonical and canonical ensembles, related mathematically to the asymptotic equivalence of conditional and exponentially-tilted measures. In this more general sense, two probability measures that are asymptotically equivalent predict the same typical or macroscopic properties of the system they are meant to model.
\end{abstract}

\pacs{%
02.50.-r, %Probability theory, stochastic processes
05.10.Gg, %Stochastic analysis methods
05.40.-a%Fluctuation phenomena, random processes, noise, and Brownian motion
}

\keywords{Equivalence of ensembles, large deviation theory, equilibrium systems, nonequilibrium systems}

\maketitle

%%% Note on definition of equivalence: not absolute continuity. Mutually absolute continuity?

%%%%%%%%%%%%%%%%%%%%%%%%%%%%%%%%%%%%%%%%%%%%%%%%%%%%%%%%%%%%%%%%%%%%%%%%%%%%%%
%%%%%%%%%%%%%%%%%%%%%%%%%%%%%%%%%%%%%%%%%%%%%%%%%%%%%%%%%%%%%%%%%%%%%%%%%%%%%%
\section{Introduction}

We study in this paper a notion of asymptotic equivalence of probability measures that generalizes the equivalence of the well-known microcanonical and canonical ensembles in the thermodynamic limit (see \cite{touchette2015} and references therein). The basic problem that we consider can be defined in a general way as follows. Let $M_n$ be a random variable defined with respect to two probability measures $P_n$ and $Q_n$ indexed by $n\in\mathbb{N}$. Can we establish conditions on these measures such that
\be
E_{P_n}[M_n]=E_{Q_n}[M_n],
\label{eqeq1}
\ee
where $E[\cin]$ denotes the expectation? For a fixed $n<\infty$, it is unlikely that such conditions exist beyond the obvious requirement that $P_n=Q_n$ almost everywhere. In the limit $n\ra\infty$, however, it is possible for two different measures to concentrate on the same value so as to give the same expectation. The aim of this paper is to put ``bounds'' on the differences between $P_n$ and $Q_n$ that guarantee that this concentration, which is related to the law of large numbers, holds for a large class of random variables. Physically, this means that two probabilistic models of a given system can predict the same typical or macroscopic properties of that system even if the models are different. 

The framework that we use to study this problem is the theory of large deviations \cite{ellis1985,dembo1998,hollander2000}. We assume that the random variable $M_n$ satisfies the large deviation principle (LDP) with respect to $P_n$ and $Q_n$ and define the set of concentration points of $M_n$ relative to either measure as the set of global minima and zeros of their respective rate function. In many applications, this set reduces to a single value, which then represents the typical value of $M_n$ (relative to $P_n$ or $Q_n$) on which its expectation concentrates exponentially as $n\ra\infty$ (again relative to $P_n$ or $Q_n$). In this context, the problem that we consider is: Under what conditions is the set of concentration points of $M_n$ relative to $P_n$ equal to the set of concentration points of $M_n$ relative to $Q_n$? In other words, under what conditions are the typical values of $M_n$ the same?

To answer these questions, we formulate in Sec.~\ref{secnot} some large deviation results related to the Radon--Nikodym derivative of $P_n$ relative to $Q_n$, which can be seen as a random variable with respect to either measure, and then use these results in Sec.~\ref{secequi} to prove essentially the following: If the Radon--Nikodym derivative is approximately equal to 1 almost everywhere, on the logarithmic scale defined by the LDP, then the two sets of concentration points of $M_n$ obtained relative to $P_n$ and $Q_n$ are the same (see the main Theorem~\ref{thmmain}). This condition on the Radon--Nikodym derivative defines, as explained in Sec.~\ref{secnot}, a general notion of asymptotic equivalence of measures from which we can summarize our main result as follows: \textit{If $P_n$ and $Q_n$ are asymptotically equivalent, then they are also equivalent at the level of typical values of $M_n$.}

This result is known to hold for specific conditional and exponentially-tilted measures, corresponding in statistical physics to the microcanonical and canonical ensembles, respectively \cite{touchette2015}. The contribution of this paper is to extend this asymptotic equivalence to a larger class of probability measures, defining general probabilistic models and stochastic processes, under precise large deviation hypotheses stated below. This extension has its source in recent works applying classical ensemble theory to describe the paths of nonequilibrium processes (see, e.g., \cite{lecomte2005,lecomte2007,garrahan2010,jack2010b,chetrite2013,chetrite2014,chetrite2015}) and relies on a special symmetry property, referred to as the fluctuation relation (see \cite{harris2007} for a review) that characterizes the fluctuations of physical quantities related to these processes. Another source is the study of random graphs, such as the Erd\"os--R\'enyi graph model and its variants, which become equivalent under some conditions in the infinite-volume limit \cite{park2004,janson2010,squartini2015,squartini2015b,garlaschelli2017}.

A formal result of Mori \cite{mori2016} pointed recently to this general equivalence for quantum systems, based on bounds on the relative entropy. The approach followed here was developed independently and is completely different: it is based on the general language of probability measures and their Radon--Nikodym derivative, and so covers both ``static'' and ``dynamic'' processes. This is illustrated in Sec.~\ref{secapps} with many applications related to sequences of random variables, equilibrium particle systems, random graphs, in addition to Markov processes evolving in discrete and continuous time. For this last application, our results provide conditions under which two stochastic processes, representing, for example, two different models for an information source or a nonequilibrium process, cannot be distinguished at the level of ergodic averages or stationary states. We also revisit in that section the equivalence of the microcanonical and canonical ensembles to clearly explain how our results extend the equivalence of classical ensembles in statistical physics.

%%%%%%%%%%%%%%%%%%%%%%%%%%%%%%%%%%%%%%%%%%%%%%%%%%%%%%%%%%%%%%%%%%%%%%%%%%%%%%
%%%%%%%%%%%%%%%%%%%%%%%%%%%%%%%%%%%%%%%%%%%%%%%%%%%%%%%%%%%%%%%%%%%%%%%%%%%%%%
\section{General framework}
\label{secnot}

%%%%%%%%%%%%%%%%%%%%%%%%%%%%%%%%%%%%%%%%%%%%%%%%%%%%%%%%%%%%%%%%%%%%%%%%%%%%%%
\subsection{Notations}

We consider two probability measures $P_n$ and $Q_n$ on a space $\Omega_n$, with $n\in\mathbb{N}$, which define technically two sequences of probability spaces. Following the introduction, we also consider a random variable $M_n:\Omega_n\ra\cM$, called a \emph{macrostate} or \emph{observable}, which is a function of the space $\Omega_n$ to a Polish space $\cM$, that is, a complete separable metric space \cite{dembo1998}.

We give examples in Sec.~\ref{secapps} of different measures and macrostates. To fix the ideas, it is useful to picture $\Omega_n$ as the space of microscopic configurations of a system of $n$ particles and $P_n$ and $Q_n$ as two probability distributions or statistical ensembles determining the likelihood of a configuration or microstate denoted by $\om=(\om_1,\om_2,\ldots,\om_n)\in\Omega_n$, where $\om_i$ is the state of the $i$th particle taking values in some set $\Omega$ so that $\Omega_n=\Omega^n$. In this case, $M_n$ could represent the total energy of the system, for example, or its magnetization if we consider a spin model. Alternatively, $\om_i\in\Omega$ could be the state of a stochastic process at time $i$, so that $\om=(\om_1,\om_2,\ldots,\om_n)$ is a path of the process from time 1 to time $n$ and $\Omega_n=\Omega^n$ is the set of all such paths. The observable $M_n$ in that case is a functional of the paths, which often takes the form of an additive or ergodic average
\be
M_n = \frac{1}{n}\sum_{i=1}^n f(\om_i),
\ee
where $f$ is some function of $\Omega$, e.g., a real-valued function, in which case $\cM$ is simply $\reals$. The measures $P_n$ and $Q_n$ then represent two different models for the stochastic process, inducing two distributions for $M_n$.

To compare these two measures, we use the \emph{Radon--Nikodym derivative} (RND) of $P_n$ relative to $Q_n$, denoted by
\be
R_n = \frac{dP_n}{dQ_n}.
\ee
This quantity establishes, as is well known, a bridge between expectations relative to $P_n$ and $Q_n$ as follows:
\be
E_{P_n}[\cin] = E_{Q_n}[R_n\cin].
\label{eqrnd1}
\ee
In particular, 
\be
P_n(B) = E_{P_n}[\idf_B]= E_{Q_n}[R_n \idf_B],
\label{eqrnd2}
\ee
where $\idf_B$ is the indicator or characteristic function of the set $B$.

The RND, as a function $R_n(\om)$ of the elements $\om\in\Omega_n$, is a real random variable having different distributions in general relative to $P_n$ and $Q_n$. To discuss the properties of these distributions, we will make the simplifying assumption throughout this paper that $P_n$ and $Q_n$ have the same support on $\Omega_n$, so that $P_n$ is absolutely continuous with respect to $Q_n$ and $Q_n$ is absolutely continuous with respect to $P_n$. In this case, $R_n$ is finite and strictly positive almost surely on the support of $P_n$ or $Q_n$. The \emph{action} $W_n$, defined by
\be
W_n = -\frac{1}{n}\log R_n,
\ee
is then also a real and finite random variable on the support of $P_n$ or $Q_n$. Up to a constant, $W_n$ is just the log-likelihood of $P_n$ relative to $Q_n$.

The reason for introducing the action is that, in many applications of interest, the RND behaves exponentially with $n$, so that its fluctuations are more conveniently studied by transforming it, as is common in large deviation theory, to a random variable whose distribution relative to $P_n$ or $Q_n$ concentrates in the limit $n\ra\infty$. The main insight needed for proving equivalence of measures is to analyze this concentration using large deviation theory. 

%%%%%%%%%%%%%%%%%%%%%%%%%%%%%%%%%%%%%%%%%%%%%%%%%%%%%%%%%%%%%%%%%%%%%%%%%%%%%%
\subsection{Large deviation principles}

The macrostate $M_n$ and the action $W_n$ are two random variables relative to $P_n$ or $Q_n$. The goal, following the introduction, is to compare the typical values of $M_n$ obtained under each measure by analyzing, via the distribution of $W_n$, the differences between these measures. The main hypothesis used to establish this comparison, which is the central hypothesis of this work, is that $M_n$ and $W_n$ jointly satisfies the large deviation principle, defined as follows. 

Let $\cY$ be a Polish space, $Y_n$ a sequence of random variables mapping $\Omega_n$ into $\cY$, $P_n$ a sequence of measures on $\Omega_n$, and $I$ a lower semi-continuous function that maps $\cY$ to $[0,\infty]$ with compact level sets. For any subset $A\subseteq\cY$, define
\be
I(A) = \inf_{y\in A} I(y).
\ee
We say that $Y_n$ satisfies the \emph{large deviation principle} (LDP) with respect to $P_n$ with rate function $I$ if
\be
\limsup_{n\ra\infty} \frac{1}{n}\log P_n(Y_n\in C)\leq -I(C)
\label{eqld1}
\ee
for any closed subset $C$ of $\cY$ and 
\be
\liminf_{n\ra\infty} \frac{1}{n}\log P_n(Y_n\in O)\geq -I(O)
\ee
for any open subset $O$ of $\cY$. The function $I(y)$, which is called the \emph{rate function}, is known to be unique and non-negative, $I\geq 0$ \cite{ellis1985,dembo1998,hollander2000}. Its \emph{domain} is the set of values $y\in\cY$ for which $I(y)<\infty$.

The LDP translates in technical terms the fact that the distribution of $Y_n$ decays exponentially in $n$, except on sets such that $I=0$. In many applications, the two large deviation bounds above are found to be the same for ``normal'' sets $A$, such as closed intervals or compact balls, which leads to
\be
\lim_{n\ra\infty}-\frac{1}{n}\log P_n(Y_n \in A) =I(A).
\ee
In the case where $\cY$ is a Euclidean space and $Y_n$ has a density $p_n(y)$ with respect to the Lebesgue measure, we can also write more simply
\be
\lim_{n\ra\infty} -\frac{1}{n}\log p_n(y) =I(y),
\ee
which clearly shows that the leading behaviour of the density of $Y_n$ is a decaying exponential in $n$, except where $I(y)=0$, with corrections in the exponential that are smaller than linear in $n$. In the large deviation and information theory literature \cite{cover1991,ellis1985,dembo1998,hollander2000}, this exponential scaling or approximation is often taken to define a \emph{logarithmic equivalence} expressed by
\be
p_n(y) \asymp e^{-nI(y)}
\ee
or
\be
P(Y_n\in A)\asymp e^{-nI(A)}.
\ee
In this sense, $a_n\asymp b_n$ means that $a_n$ and $b_n$ are equal up to $e^{o(n)}$ corrections in $n$ or, more precisely,
\be
\lim_{n\ra\infty} \frac{1}{n}\log \frac{a_n}{b_n} =0.
\label{eqls1}
\ee

With these definitions, we express our main hypotheses as follows. 
\begin{hypothesis}~
\label{hyp1}
\begin{itemize}
\item The couple $(M_n,W_n)$ satisfies, as a random variable on the product space $\cM\times\reals$, the LDP relative to $P_n$ with joint rate function $K_P$;
\item $(M_n,W_n)$ satisfies the LDP relative to $Q_n$ with joint rate function $K_Q$;
\item $K_P$ and $K_Q$ have the same domain.
\end{itemize}
\end{hypothesis}

These hypotheses are satisfied in many applications. The first one means essentially that
\be
p_n(M_n=m,W_n= w) \asymp e^{-n K_P(m,w)},
\ee
assuming formally that the joint probability density of $M_n$ and $W_n$ exists. A similar result holds for $Q_n$ with the rate function $K_Q$. In the absence of densities, the meaning of the LDP is as defined above with the upper and lower bounds. In all cases, our prior assumption that $P_n$ and $Q_n$ have the same support is reflected in the hypothesis that $K_P$ and $K_Q$ have the same domain.

In general, it is known that having the joint LDP for two random variables implies that each random variable also satisfies the LDP. This marginalization of the LDP can be derived from the definition of this principle or from the so-called contraction principle \cite[Thm.~4.2.1]{dembo1998}, and leads to variational formula for the marginal rate functions of $M_n$ and $W_n$.

\begin{prop}
\label{propcont}
Under Hypotheses~\ref{hyp1}, $M_n$ satisfies the LDP relative to $P_n$ with marginal rate function
\be
J_P(m)=\inf_{w\in\reals} K_P(m,w)
\label{eqcon1}
\ee
and the LDP relative to $Q_n$ with marginal rate function
\be
J_Q(m) =\inf_{w\in\reals} K_Q(m,w).
\label{eqcon2}
\ee
which has the same domain as $J_P$. Similarly, $W_n$ satisfies the LDP relative to $P_n$ and $Q_n$ with rate functions
\be
I_P(w) = \inf_{m\in\cM} K_P(m,w)
\label{eqcon3}
\ee
and
\be
I_Q(w)=\inf_{m\in\cM} K_Q(m,w),
\label{eqcon4}
\ee
respectively, having the same domain.
\end{prop}

These formul\ae\ can be justified easily in terms of densities by applying the LDP and the Laplace principle for approximating exponential integrals. Considering, for example, the marginalization of $W_n$ in
\be
p_n(M_n=m) = \int_\reals  p_n(m,w)dw\asymp\int_\reals  e^{-nK_P(m,w)}dw
\ee
leads to
\be
p_n(M_n= m) \asymp\exp\left(-n\inf_{w\in\reals} K_P(m,w)\right)=\exp\left(-n J_P(m)\right)
\ee
in the limit $n\ra\infty$. 

We give next a rigorous proof for measures based on the contraction principle of large deviation theory \cite{dembo1998}, which is itself an application of the Laplace principle \cite{hollander2000,touchette2009}. 

\begin{proof}
The contraction principle states that, if $Y_n$ satisfies the LDP with rate function $I$, then $Z_n=f(Y_n)$ satisfies the LDP with rate function
\be
J(z) = \inf_{y:f(y)=z} I(y)
\ee
if the ``contraction'' function $f$ is continuous \cite[Thm.~4.2.1]{dembo1998}.

In the case of marginalizing, for example, from $(M_n,W_n)$ to $M_n$, the contraction function is simply a projection $f(M_n,W_n)=M_n$, which is continuous under the natural product topology for the space of $(M_n, W_n)$. Therefore, 
\be
J_P(m)=\inf_{w\in\reals:f(m,w)=m}K_P(m,w)=\inf_{w\in\reals} K_P(m,w).
\ee
All other contractions follow in the same way. Moreover, the fact that the marginal rate functions have the same domain simply follows from our assumption that the joint rate functions have the same domain. 
\end{proof}

Many techniques can be used to derive the LDP for $M_n$ and $W_n$ and the corresponding rate function, though the derivation of LDPs is, as always, a difficult problem. In the case of equilibrium many-particle systems, one can use the contraction principle for observables $M_n$ that admit a representation function, as described in Sec.~5.3.4 of \cite{touchette2009}, or contractions based on the so-called \emph{level 3} of large deviations \cite{ellis1985}, which are however very difficult to work with. For Markov processes, one can also use the contraction principle when considering observables that depend on both the state of the process and its jumps or increments \cite{chetrite2014}. In this case, the contraction is applied to the \emph{level 2.5} of large deviations, which involves explicit LDPs for the empirical measure and empirical current \cite{barato2015,bertini2015,hoppenau2016}. Finally, when $M_n$ takes values in $\reals^d$ one can use the G\"artner--Ellis Theorem, which is based on the following function:
\be
\lambda_P(k,\eta) =\lim_{n\ra\infty}\frac{1}{n}\log E_{P_n} [e^{n\langle k, M_n\rangle+n\eta W_n}],
\label{eqscgfact1}
\ee
called the \emph{scaled cumulant generating function}. Here $k\in\reals^d$, $\langle\cdot,\cdot\rangle$ is the standard scalar product, and $\eta\in\reals$. Provided that this function exists in an open neighbourhood of the origin and is ``steep'' (see \cite{ellis1985,dembo1998,hollander2000} for details), this theorem states that $(M_n,W_n)$ satisfies the LDP with rate function $K_P$ given by the Legendre--Fenchel transform of $\lambda_P$:
\be
K_P(m,w) = \sup_{k\in\reals^d,\eta\in\reals}\{\langle k,m\rangle+\eta w-\lambda_P(k,\eta)\}.
\ee 
For independent and identically distributed random variables, $\lambda_P$ reduces to a simple cumulant function, while for Markov processes it is given by the dominant eigenvalue of a matrix or linear operator \cite{deuschel1989,dembo1998,touchette2017}. 

%%%%%%%%%%%%%%%%%%%%%%%%%%%%%%%%%%%%%%%%%%%%%%%%%%%%%%%%%%%%%%%%%%%%%%%%%%%%%%
\subsection{Typical sets and values}

Since rate functions are non-negative, we have as a consequence of Prop.~\ref{propcont} that, if $(m^*,w^*)$ is a zero of $K_P$, then $m^*$ must be a zero of $J_P(m)$ and $w^*$ must be a zero of $I_P(w)$. A similar result holds relative to $Q_n$.

The zeros of rate functions will play an important role in the remaining, so it is important to discuss their interpretation. To this end, let us consider the rate function $J_P$ describing the large deviations of $M_n$ relative to $P_n$, and let $\cE_P$ denote the set of zeros of $J_P$, which also corresponds to the set of global minima of $J_P$:
\be
\cE_P=\{m\in\cM: J_P(m)=0\}.
\ee
Because $J_P$ has compact level sets, $\cE_P$ is compact and non-empty. 

In general, $\cE_P$ represents the typical set on which the distribution of $M_n$ concentrates in the limit $n\ra\infty$. To be more precise, it can be proved (see \cite[Thm.~2.5]{ellis2000}) that the sequence $P_n(M_n\in\cin)$ converges weakly to a probability measure $\Pi$ on $\cM$ such that $\Pi(\cE_P)=1$. This follows because the probability of any point that is not in $\cE_P$ decays exponentially as a result of the LDP, so that $P_n(M_n\in\cin)$ must concentrate on $\cE_P$ as $n\ra\infty$. For this reason, $\cE_P$ is called the \emph{concentration set} or the \emph{typical set} of $M_n$ relative to $P_n$. 

If $J_P$ has a unique minimum and zero $ m^*$, then the sequence $P_n(M_n\in\cin)$ converges weakly to the delta measure $\delta_{m^*}$ \cite[Thm.~2.5]{ellis2000}, so that $m^*$ is the unique \emph{concentration} or \emph{typical value} of $M_n$. In this case, $M_n$ satisfies a weak law of large numbers in the sense that 
\be
\lim_{n\ra\infty} P_n(\|M_n-m^*\|> \eps)=0,
\label{eqlln1}
\ee
where $\eps$ is any positive real number and $\|\cdot\|$ is a metric on $\cM$. We then also say that $M_n\ra m^*$ in probability (relative here to $P_n$).

These notions of typical sets and values can be applied to any of the rate functions defined before. In applications, it is more common to find that a random variable satisfying the LDP has a unique typical value than a ``extended'' typical set, so we focus here mainly on the former type of concentration. In general, a random variable has a unique typical value if its rate function is strictly convex. 

%%%%%%%%%%%%%%%%%%%%%%%%%%%%%%%%%%%%%%%%%%%%%%%%%%%%%%%%%%%%%%%%%%%%%%%%%%%%%%
\subsection{Fluctuation relations}

The different rate functions defined up to now are not independent, since probabilities obtained with $P_n$ can be expressed, as shown in (\ref{eqrnd2}), as modified expectations with respect to $Q_n$ that involve the RND. This leads, as shown next, to a simple relation between the rate functions involving the action, referred to in statistical physics as \emph{fluctuation relations} \cite{lebowitz1999}.

\begin{prop}
\label{propfr}
The joint rate functions $K_P$ and $K_Q$ of $(M_n,W_n)$ are related by
\be
K_P(m,w) = w+K_Q(m,w)
\label{eqfr1}
\ee
for all $m\in\cM$ and $w\in\reals$. Similarly, the marginal rate functions $I_P$ and $I_Q$ of $W_n$ satisfy
\be
I_P(w) = w+I_Q(w)
\label{eqfr2}
\ee
for all $w\in\reals$.
\end{prop}

This result is obvious if we assume again that densities exist. Then the Radon--Nikodym formula (\ref{eqrnd2}) simply becomes
\be
p_n(m,w)=E_{Q_n}[e^{-nW_n}\delta(M_n-m)\delta(W_n-w)]=e^{-nw}q_n(m,w).
\ee

The proof next translates this observation for measures using another important result of large deviation theory known as Varadhan's Lemma. We refer to \cite[Thm.~1.3.4]{dupuis1997} or \cite[Thm.~4.3.1]{dembo1998} for the general formulation of this result. 

\begin{proof}
The probability measure
\be
P_n(M_n\in A,W_n\in B)= \int_A\int_B  P_n(dm,dw)
\ee 
is equivalent, using the Radon--Nikodym formula (\ref{eqrnd2}), to
\be
P_n(M_n\in A,W_n\in B) = \int_A\int_B e^{-nw} Q_n(dm,dw).
\label{eqvari1}
\ee
This has the form of an exponential integral with $X^n=(M_n,W_n)$ and $h(x)=-w$ in the notations of Theorem~1.3.4 of \cite{dupuis1997}. The function $h$ in our case is not bounded. However, since $R_n$ is strictly positive on the support of $P_n$, the large deviation upper bound
\be
\limsup_{n\ra\infty} \frac{1}{n}\log P_n(M_n\in A,W_n\leq -C)\leq -\inf_{w\leq -C}K_P(A,w)
\ee
implies
\be
\lim_{C\ra\infty }\limsup_{n\ra\infty} \frac{1}{n}\log P_n(M_n\in A,W_n\leq -C)=-\infty
\ee
for any measurable $A$. Therefore, the technical condition stated in \cite[Thm.~1.3.4]{dupuis1997} is satisfied, leading to the main result
\be
\lim_{n\ra\infty} -\frac{1}{n}\log P_n(M_n\in A,W_n\in B)=\inf_{m\in A,w\in B} \{w+K_Q(m,w)\}.
\ee
Since rate functions are unique \cite[Lem.~4.1.4]{dembo1998}, the right-hand side must be the rate function of $(M_n,W_n)$ relative to $P_n$, which proves (\ref{eqfr1}).

The same reasoning applied to $P_n(W_n\in A)$ yields (\ref{eqfr2}). Alternatively, we can derive (\ref{eqfr2}) more directly by applying the contraction principle to marginalize $M_n$ from (\ref{eqfr1}) following Prop.~\ref{propcont}.
\end{proof}

The relations (\ref{eqfr1}) and (\ref{eqfr2}) are interpreted in statistical physics as symmetries on rate functions that impose general constraints on the fluctuations of nonequilibrium processes (see \cite{harris2007} for a review). In this context, $P_n$ refers to the probability measure of a stationary Markov process modelling a nonequilibrium process, $Q_n$ is the probability measure of the time-reversed process, and $W_n$ is then called the entropy production. We will come back to this example in Sec.~\ref{secapps}.

%%%%%%%%%%%%%%%%%%%%%%%%%%%%%%%%%%%%%%%%%%%%%%%%%%%%%%%%%%%%%%%%%%%%%%%%%%%%%%
%%%%%%%%%%%%%%%%%%%%%%%%%%%%%%%%%%%%%%%%%%%%%%%%%%%%%%%%%%%%%%%%%%%%%%%%%%%%%%
\section{Concentration equivalence}
\label{secequi}

We are now ready to prove the equivalence of $P_n$ and $Q_n$ at the level of the typical sets of $M_n$ defined, respectively, as
\be
\cE_P = \{m\in\cM:J_P(m)=0\}
\ee
and
\be
\cE_Q = \{m\in\cM:J_Q(m)=0\}.
\ee
Since $J_P$ and $J_Q$ have compact level sets, $\cE_P$ and $\cE_Q$ are non-empty and compact.

The basic idea for proving this equivalence is contained in the fluctuation symmetry (\ref{eqfr1}), which shows that the rate function $K_P(m,w)$ and $K_Q(m,w)$ can vanish on the same value $m$ if they vanish for $w=0$. To prove that $\cE_P=\cE_Q$, we then need to make sure that $w=0$ is the only value where these rate functions vanish, so that $W_n$ has a unique typical value equal to $0$ with respect to both $P_n$ and $Q_n$. 

As a result, we assume from now on that the rate functions $I_P$ and $I_Q$ of $W_n$ each have a unique zero, which is not necessarily equal to $0$, and define the following. We say that \emph{$P_n$ and $Q_n$ are asymptotically equivalent if}
\be
\lim_{n\ra\infty} \frac{1}{n}\log \frac{dP_n}{dQ_n}=0
\label{eqdefae}
\ee
\emph{in probability with respect to $P_n$ and $Q_n$}. Note that this definition is consistent with the symmetry (\ref{eqfr2}), for if $I_Q(0)=0$ then $I_P(0)=0$, and vice versa.

The next theorem, which is the main result of this paper, shows that this notion of asymptotic equivalence of measures is sufficient for $\cE_P$ to coincide with $\cE_Q$. 

\begin{thm}
\label{thmmain}
Assume that $M_n$ and $W_n$ satisfy the joint LDP stated in the Hypotheses~\ref{hyp1} and that the rate functions $I_P$ and $I_Q$ of $W_n$ each have a unique zero. If $P_n$ and $Q_n$ are asymptotically equivalent, then $\cE_P=\cE_Q$.
\end{thm}

\begin{proof}
The assumption that $I_P$ and $I_Q$ have unique zeros, coupled with the assumption that $P_n$ is asymptotically equivalent to $Q_n$, implies that $I_P(0)=I_Q(0)=0$ and that $w=0$ is the only point where this equality holds.

The equality $I_P(0)=0$ leads with (\ref{eqcon3}) to
\be
0=I_P(0) = \inf_{m\in\cM} K_P(m,0).
\ee
Let $A$ denote the set of minimizers of the infimum over $m$. Then $K_P(m^*,0)=0$ where $m^*\in A$ and, from (\ref{eqcon1}), we obtain
\be
J_P(m^*) = \inf_{w\in\reals} K_P(m^*,w)=K_P(m^*,0)=0,
\label{eqp1}
\ee
which implies that $m^*\in\cE_P$.

By applying the symmetry (\ref{eqfr1}), we also have $K_Q(m^*,0)=0$ and so
\be
J_Q(m^*)=\inf_{w\in\reals} K_Q(m^*,w) = K_Q(m^*,0)=0,
\ee 
which implies that $m^*\in\cE_Q$.

This only shows that all $m^*\in A$ are in $\cE_P$ and in $\cE_Q$ or, equivalently, that $A\subset\cE_P$ and $A\subset\cE_Q$. To prove that all $m\in\cE_P$ are in fact in $A$, assume that $\bar m\in\cE_P$ and that the infimum over $w$ in (\ref{eqp1}) is achieved at $0$. Then
\be
0=J_P(\bar m) = K_P(\bar m, 0)
\ee 
so that $\bar m\in A$. On the other hand, if the infimum is achieved for $\bar w\neq 0$, then
\be
I_P(\bar w) = \inf_{m\in\cM} K_P(m,\bar w) = K_P(\bar m,\bar w)=0,
\ee
which would contradict the fact that $w=0$ is the only zero of $I_P(w)$. Consequently, we have proved that $A=\cE_P$ and so that $\cE_P\subset \cE_Q$.

To prove the equality of the two sets, we only have to use the same argument by starting with $Q_n$ to show similarly that all $m\in\cE_Q$ are also in $\cE_P$, so that $\cE_Q\subset \cE_P$. Consequently, $\cE_P=\cE_Q$.
\end{proof}

The result of Theorem~\ref{thmmain} is natural considering that the notion of asymptotic equivalence and the LDP are based on the same logarithmic scale ($\asymp$), defined in (\ref{eqls1}), so that differences between measures that are neglected on that scale should not affect the LDP of $M_n$. One has to be careful with this intuition, however, because it is known that sub-exponential differences between $P_n$ and $Q_n$ can lead to different rate functions \cite{touchette2015}. What Theorem~\ref{thmmain} shows is that such differences do not influence the \emph{concentration} of $M_n$, although they can influence the \emph{fluctuations} of $M_n$. In other words, if $P_n$ and $Q_n$ are asymptotically equivalent, then the rate functions $J_P$ and $I_Q$ for $M_n$ are not necessarily equal, but they have the same zeros. 

The next result relates the notion of asymptotic equivalence, defined in (\ref{eqdefae}) in terms of $W_n$, to the relative entropy
\be
D(P_n||Q_n)=\int dP_n \log\frac{dP_n}{dQ_n} =E_{P_n}\left[\log\frac{dP_n}{dQ_n}\right]
\ee 
or Kullback--Leibler distance \cite{cover1991}. This result is potentially useful for determining whether $P_n$ and $Q_n$ are asymptotically equivalent without having to explicitly derive the rate function of $W_n$.

\begin{thm}
\label{thmrelent}
Assume the same hypotheses as in Theorem~\ref{thmmain}. If $P_n$ and $Q_n$ are asymptotically equivalent, then 
\be
\lim_{n\ra\infty} \frac{1}{n}D(P_n||Q_n) = \lim_{n\ra\infty} \frac{1}{n}D(Q_n||P_n) = 0.
\label{eqlimre1}
\ee
Conversely, if the limits above hold, then $P_n$ and $Q_n$ are asymptotically equivalent.

\end{thm}

\begin{proof}
The proof only relies on the law of large numbers for $W_n$. If $W_n$ satisfies the LDP relative to $P_n$ and its rate function $I_P$ has a unique zero $w^*$, as assumed, then
\be
\lim_{n\ra\infty} E_{P_n}[W_n] =w^*.
\label{eqmean1}
\ee
Therefore, if $P_n$ and $Q_n$ are asymptotically equivalent, then $w^*=0$ and
\be
\lim_{n\ra\infty} E_{P_n}[W_n]  =\lim_{n\ra\infty}-\frac{1}{n}E_{P_n}\left[\log\frac{dP_n}{dQ_n}\right]=\lim_{n\ra\infty} -\frac{1}{n} D(P_n||Q_n) = 0.
\ee
The same applies relative to $Q_n$.

To prove the converse, note that if the limit (\ref{eqmean1}) in mean applies for $W_n$ and $I_P$ has a unique zero, as assumed, then the limiting mean $w^*$ must be that zero. Hence, if the first limit for the relative entropy shown in (\ref{eqlimre1}) holds, then $I_P(0)=0$. Since the same applies for $I_Q$, we conclude that $P_n$ and $Q_n$ are asymptotically equivalent.
\end{proof}

\begin{remark*}\
\begin{enumerate}

\item The result of Theorem~\ref{thmmain} was already known to hold, as mentioned in the introduction, for the specific probability measures that are the microcanonical and canonical ensembles of statistical physics (see Sec.~\ref{secapps}). The notion of asymptotic equivalence of measures used here comes from that context.

\item The equivalence result is valid for \emph{any} random variable $M_n$ that satisfies the joint LDP with $W_n$. This means concretely that, if two probabilistic models of some system are asymptotically equivalent, then they are indistinguishable at the level of their typical sets. They are equivalent models predicting the same typical properties.

\item\label{rem3} The asymptotic equivalence of $P_n$ and $Q_n$ is a sufficient but not a necessary condition for $\cE_P=\cE_Q$. In some cases (see Sec.~\ref{secapps}), we can indeed have $\cE_P=\cE_Q$ for specific random variables $M_n$ even though $P_n$ and $Q_n$ are not asymptotically equivalent.

\item The notion of asymptotic equivalence is transitive: If $P_n$ is asymptotically equivalent to $Q_n$ and $Q_n$ is asymptotically equivalent to $F_n$, then $P_n$ is asymptotically equivalent to $F_n$. This can be checked directly from the definition of asymptotic equivalence.

\item\label{rem5} When the limits (\ref{eqlimre1}) for the relative entropy hold, $P_n$ and $Q_n$ are said to have zero divergence rate \cite{shields1993} or to be equivalent in the specific relative entropy sense \cite{lewis1994a,lewis1994,lewis1995,touchette2015}. For Markov processes, the action and relative entropy can be related to transition and waiting times \cite{chazottes2006}.

\item We need not assume for proving Theorem~\ref{thmmain} that the rate function $J_P$ and $J_Q$ defining the typical sets $\cE_P$ and $\cE_Q$ have unique zeros. This assumption is only required for $I_P$ and $I_Q$ so as to have unique typical values for $W_n$ relative to $P_n$ and $Q_n$ which, by the assumption of asymptotic equivalence,  are equal to $0$.

\item An open problem is to determine what happens when $W_n$ has another typical value other than $0$ or when $w=0$ is only in the typical set of $P_n$ or $Q_n$ without being a real concentration value. The proof given here suggests that $\cE_P$ and $\cE_Q$ should have in this case some overlap without being equal, as is known to happen for the microcanonical and canonical ensembles when they are \emph{partially equivalent} \cite{ellis2000}.

\item Another open problem is to generalize our results when $P_n$ and $Q_n$ do not have the same support. In this case, the symmetry relations expressed in Prop.~\ref{propfr} do not seem to hold on the whole domain of the rate functions involved, but only on their intersection. It is not clear then whether or not this is enough to have equivalence of typical sets, as there is no guarantee that the zeros of the rate functions relative to $P_n$ are also zeros of the rate functions relative to $Q_n$, which makes their comparison more complicated. The microcanonical and canonical ensembles, which do not have the same support, should serve as a starting point for understanding this problem.
\end{enumerate}
\end{remark*}

%%%%%%%%%%%%%%%%%%%%%%%%%%%%%%%%%%%%%%%%%%%%%%%%%%%%%%%%%%%%%%%%%%%%%%%%%%%%%%
%%%%%%%%%%%%%%%%%%%%%%%%%%%%%%%%%%%%%%%%%%%%%%%%%%%%%%%%%%%%%%%%%%%%%%%%%%%%%%
\section{Applications}
\label{secapps}

We illustrate in this section the result of Theorem~\ref{thmmain} using various examples of probability measures and stochastic processes. The examples are simple: they are presented to discuss certain aspects of that theorem and to give an idea of how it can be applied to measures that describe a wide range of ``static'' and ``dynamic'' probabilistic models.

\subsection{Independent random variables}

We first consider a sequence $X_1,X_2,\ldots,X_n$ of real random variables, assumed to be independent and identically distributed (iid) according to some density $p$, defining our model $P_n$, or the density $q$, defining $Q_n$. For example, we can choose $p\sim \cN(0,1)$ to be a standard normal random variables and $q\sim\cN(\mu,\sigma^2)$ to be a Gaussian random variable with mean $\mu$ and variance $\sigma^2$. In this case, the RND is simply
\be
R_n = \prod_{i=1}^n \frac{p(X_i)}{q(X_i)}=e^{-nW_n},
\ee
where
\be
W_n = \frac{\mu}{\sigma^2}M_n +\frac{(\sigma^2-1)}{2\sigma^2}C_n-\frac{\mu^2}{2\sigma^2}-\log \sigma
\ee
with
\be
M_n = \frac{1}{n}\sum_{i=1}^n X_i,\qquad C_n = \frac{1}{n}\sum_{i=1}^n X_i^2.
\label{eqobs1}
\ee

To determine whether $P_n$ and $Q_n$ are asymptotically equivalent, we need to find the rate functions $I_P$ and $I_Q$ of the action $W_n$. This can be done easily with the G\"artner--Ellis Theorem (see Sec.~\ref{secnot}) or by contraction of the joint LDP of $M_n$ and $C_n$ above. From the form of $W_n$, however, it is clear that $P_n$ and $Q_n$ are asymptotically equivalent if and only if $\mu=0$ and $\sigma=1$, that is, if and only if we trivially have $p=q$. In this case, $W_n=0$ with probability 1, so that $I_P$ and $I_Q$ are degenerate on $w=0$.

For $\mu=0$ and $\sigma\neq 1$, $M_n\ra 0$ in probability relative to both $P_n$ and $Q_n$, although the two measures are not asymptotically equivalent. For this observable, we therefore have $\cE_P=\cE_Q$, which shows that the condition of asymptotic equivalence is not a necessary condition for the equivalence of $\cE_P$ and $\cE_Q$, as noted in Remark~\ref{rem3}. Note, however, that $\cE_P\neq \cE_Q$ if we take the observable to be $C_n$, since $C_n\ra 1$ in probability relative to $P_n$ while $C_n\ra \sigma^2$ in probability relative to $Q_n$, assuming again $\mu=0$. This suggests that, if $P_n$ and $Q_n$ are not asymptotically equivalent, then there is at least one observable for which $\cE_P\neq\cE_Q$, a result that would be interesting to prove in general.

The asymptotic equivalence obtained for $p=q$ applies in a more general way to any iid sequences satisfying the hypotheses of this work. This follows from Theorem~\ref{thmrelent} by noting that 
\be
\lim_{n\ra\infty}\frac{1}{n}D(P_n||Q_n)  = D(p||q) = \int dx\, p(x)\log\frac{p(x)}{q(x)}
\ee
and that $D(p||q)$ vanishes if and only if $p(x)=q(x)$ almost everywhere \cite{cover1991}. 

To go beyond this trivial case of equivalence, we can consider sequences of random variables that are independent but not identically distributed. In particular, we can consider in $Q_n$ all but one random variable, say $X_1$, to have the same distribution $p$, so that 
\be
R_n(x_1,x_2,\ldots, x_n) = \frac{p(x_1) p(x_2)\cdots p(x_n)}{q(x_1) p(x_2)\cdots p(x_n)}= \frac{p(x_1)}{q(x_1)}
\ee
and thus
\be
W_n = \frac{1}{n}\log \frac{q(X_1)}{p(X_1)}.
\ee 
In this case, $W_n\ra 0$ as $n\ra\infty$ relative to both $P_n$ and $Q_n$, provided that $p$ and $q$ do not scale with $n$ and have the same support. Under these additional conditions, $P_n$ and $Q_n$ are then asymptotically equivalent. This can be generalized, as is clear from the form of $R_n$ above, to cases where a number $N<n$ of independent random variables have a different distribution $q$ under $Q_n$, so long as $N/n\ra 0$ as $n\ra \infty$. 

\subsection{Microcanonical and canonical ensembles}

The microcanonical and canonical ensembles are the main probabilistic models used in statistical physics to study equilibrium systems. Both are defined by transforming a basic measure $\mu_n$ on the space $\Omega_n$ of configurations or microstates of a system of $n$ particles, for which the random variable $M_n$ is interpreted as a macrostate. On the one hand, the \emph{microcanonical ensemble} is the measure on $\Omega_n$ obtained by conditioning $\mu_n$ on $M_n\in B$:
\be
\mu_n(d\om|M_n\in B) =\frac{\mu_n(d\om,M_n\in B)}{\mu_n (M_n\in B)}
=
\left\{
\begin{array}{lll}
\mu_n (d\om)/\mu_n(B) & & \text{if }M_n(\om)\in B\\
0 & & \text{otherwise,}
\end{array}
\right.
\label{eqmicro1}
\ee
where $\om$ is an element of $\Omega_n$. Usually, $M_n\in\reals$ is the energy of the system and $B$ is a very thin interval $[\barm-\eps,\barm+\eps]$, called the energy shell, located around a fixed value $\barm$. Taking this to represent our model $P_n$, we then have
\be
P_n(d\om) =\mu_n(d\om|M_n\in [\barm-\eps,\barm+\eps]).
\ee
On the other hand, the \emph{canonical ensemble} is the measure on $\Omega_n$ that transforms $\mu_n$ according to
\be
Q_n(d\om)=\frac{e^{nk M_n(\om)}}{E_{\mu_n}[e^{nk M_n}]}\mu_n(d\om),\quad k\in\reals
\label{eqcan1}
\ee
provided that $E_{\mu_n}[e^{nk M_n}]<\infty$. This measure is also called the exponential-tilting of $\mu_n$ or the exponential family, and represents physically the distribution of a system of $n$ particles with energy $M_n$ in contact with a heat bath at inverse temperature $\beta=-k$. Mathematically, it also represents a ``softening'' of the microcanonical measure in which the ``hard'' conditioning constraint $M_n=\barm$ is replaced by the ``soft'' constraint $E_{Q_n}[M_n]=\barm$ on the average of $M_n$.

To prove the equivalence of these two measures, we need to assume that $M_n$ satisfies the LDP with respect to $\mu_n$ with rate function $I$. Assuming that $I$ is convex at $\barm$ and choosing $k\in\partial I(\barm)$, where $\partial I$ denotes the sub-differential of $I$ \cite{rockafellar1970}, it can be shown that $W_n\ra 0$ in probability relative to both $P_n$ and $Q_n$ \cite{touchette2015}. The two measures or ensembles must then be equivalent at the level of typical sets of random variables that satisfy the LDP in both ensembles. The full proof of this result can be found in \cite{touchette2015}, so we do not repeat it here.

Physically, nonequivalent ensembles arise when the interactions between particles in a macroscopic system are long-range, with mean-field interactions being an extreme case of long-range interactions. For examples of such systems, see \cite{ellis2000,campa2009,touchette2015}. When the interaction is short or finite range, the microcanonical and canonical ensembles are generally equivalent.

The equivalence of the microcanonical and canonical ensembles has also been investigated recently in the context of random graphs \cite{squartini2015,squartini2015b,garlaschelli2017}, sometimes with different notions of asymptotic equivalence \cite{janson2010}. What is found in general is that a microcanonical ensemble of random graphs in which a \emph{fixed} number of constraints are considered is equivalent to a canonical ensemble of graphs in which these constraints are imposed on average with an exponential (canonical) tilting. One example is the Erd\"os--R\'enyi ensemble in which all graphs with $N$ nodes and $E$ links have the same probability, and $E$ is such that the degree per node $2E/N$ converges to a constant $d$ as $N\ra\infty$ (sparse regime). In the limit where $N\ra\infty$, this microcanonical graph ensemble is known to be equivalent with the more common canonical ensemble in which the $N$ vertices are linked at random with probability $p=d/N$. However, if an \emph{extensive} number of constraints proportional to the number of nodes are imposed, then the microcanonical and canonical ensemble can be nonequivalent. This is illustrated in \cite{squartini2015} with random graphs in which the whole degree sequence is fixed.

\subsection{Generalized canonical ensembles}

The canonical ensemble is not the only probability measure that is asymptotically equivalent to the microcanonical ensemble. More generally, we can replace the canonical measure $Q_n$ in (\ref{eqcan1}) by
\be
F_n(d\om)=\frac{e^{nh(M_n(\om))}}{E_{\mu_n}[e^{nh(M_n)}]}\mu_n(d\om),
\label{eqgencan1}
\ee
where $h:\cM\ra\reals$ is a real function of $M_n$ such that $E_{\mu_n}[e^{nh(M_n)}]<\infty$. This defines in statistical physics a \emph{generalized canonical ensemble} \cite{costeniuc2005,costeniuc2006,costeniuc2006b}, which has the same support as the canonical ensemble measure (\ref{eqcan1}) and which can be made equivalent to both the canonical ensemble and the microcanonical ensemble.

The asymptotic equivalence with the microcanonical ensemble is discussed in detail in \cite{costeniuc2005}. To see how the generalized canonical ensemble can be equivalent with the standard canonical ensemble, let us assume as before that $M_n$ satisfies the LDP relative to $\mu_n$ with rate function $I$ and define
\be
\phi(h)=\lim_{n\ra\infty}\frac{1}{n}\log E_{\mu_n}[e^{nh(M_n)}]
\ee
and
\be
\lambda(k)=\lim_{n\ra\infty}\frac{1}{n}\log E_{\mu_n}[e^{nkM_n}],
\ee
assuming that both are finite. By Varadhan's Lemma \cite[Thm.~1.3.4]{dupuis1997}, it is known that these two functions can be expressed in terms of the rate function $I$ as
\be
\phi(h)=\sup_{m\in\cM}\{h(m)-I(m)\}
\ee
and
\be
\lambda(k) = \sup_{m\in\cM}\{km -I(m)\}.
\ee
Moreover, under the hypothesis of this lemma, it can be proved (see \cite[Thm.~11.7.2]{ellis1985} or \cite[Thm.~14]{touchette2015}) that $M_n$ satisfies the LDP relative to $F_n$ with rate function
\be
I_F(m) = I(m)-h(m)+\phi(h)
\ee
and the LDP relative to $Q_n$ (the canonical ensemble) with rate function 
\be
I_Q(m) = I(m) -km+\lambda(k).
\ee

Combining these results, we see that, if $h$ and $k$ are chosen such that $I_F(m)$ and $I_Q(m)$ have the same unique minimum and zero $\barm$, then $\phi(h)=h(\barm)-I(\barm)$ and $\lambda(k)=k\barm-I(\barm)$. Consequently,
\be
\lim_{n\ra\infty} \frac{1}{n}\log\frac{dF_n}{dQ_n} = h(\barm)-k \barm+\lambda(k)-\phi(h) =0
\ee
in probability relative to both $F_n$ and $Q_n$, which means that we have asymptotic equivalence. This follows here because the RND is a function of $M_n$ only and both $F_n$ and $Q_n$ concentrate on the same value $\barm$ of $M_n$. 

Specific examples of generalized ensembles related to long-range and mean-field interacting systems are discussed in \cite{costeniuc2005,costeniuc2006,costeniuc2006b,touchette2010b}. The advantage of using the generalized canonical ensemble is that it can be used to describe the microcanonical properties of many-body systems whenever the canonical ensemble itself is not equivalent with the microcanonical ensemble. This happens generically when the entropy is nonconcave as a function of the energy in the thermodynamic limit. For more details on equivalent versus nonequivalent ensembles, we refer to \cite{touchette2004b,touchette2015}. 

\subsection{Markov processes}

We close the list of examples by briefly discussing Markov processes, beginning with the case of Markov chains. 

Let $X_1,X_2,\ldots, X_n$ be an ergodic Markov chain on a set $\Omega$, assumed to be finite for simplicity, and consider two probability measures $P_n$ and $Q_n$ on the space $\Omega_n=\Omega^n$ defined by the (homogeneous) transition kernels $p(x,y)$ and $q(x,y)$, respectively. Starting with the same distribution $\rho$ for $X_1$, we thus write
\be
P_n(x_1,x_2,\ldots, x_n) = \rho(x_1) p(x_1,x_2) \cdots p(x_{n-1},x_n)
\ee
and
\be
Q_n(x_1,x_2,\ldots, x_n) = \rho(x_1) q(x_1,x_2)\cdots q(x_{n-1},x_n),
\ee
so that
\be
W_n = \frac{1}{n}\sum_{i=1}^{n-1} \log \frac{q(x_i,x_{i+1})}{p(x_i,x_{i+1})}.
\ee
The rate function of $W_n$, if it exists, can be derived by contracting the LDP of the so-called pair empirical distribution of the Markov chain; see Sec.~4.3 of \cite{touchette2009}. Alternatively, we can notice that the relative entropy rate of the two Markov chains is
\be
\lim_{n\ra\infty}\frac{1}{n}D(P_n||Q_n) = \sum_{(x,y)\in \Omega^2} \mu(x) p(x,y) \log\frac{p(x,y)}{q(x,y)}
\ee
where $\mu(x)$ is the invariant distribution of the Markov chain with transition distribution $p(x,y)$ \cite{cover1991}. Since the relative entropy on the right-hand side above vanishes if and only if $p=q$ almost everywhere, we then obtain, similarly to iid sequences, that $P_n$ and $Q_n$ are asymptotically equivalent, under the conditions of Theorem~\ref{thmmain}, if they define the same Markov chain with the same transition kernel. This applies to homogeneous Markov chains. As in the case of iid sequences, there is more room for equivalence if we allow the transition kernels to be time-dependent or compare Markovian with non-Markovian processes. 

Similar results can be formulated for Markov chains on uncountable and continuous spaces, provided that they have the LDPs required in Hypotheses~\ref{hyp1}. One can also consider continuous-time processes, such as pure diffusions, by replacing $\Omega_n$ with the space $\Omega_T$ of sample paths over the time interval $[0,T]$, in which case $P_n$ and $Q_n$ are ``path'' measures similar to the Wiener measure, denoted by $P_T$ and $Q_T$, whose action
\be
W_T = -\frac{1}{T}\log\frac{dP_T}{dQ_T}
\ee
can be expressed in terms of stochastic integrals using Girsanov's Theorem \cite{pavliotis2014}. The large deviation limit defining the equivalence of $P_T$ and $Q_T$ is then the long-time or ergodic limit $T\ra\infty$. 

Many examples of stochastic processes related to nonequilibrium systems which are asymptotically equivalent are treated in \cite{chetrite2014}. This work introduced together with \cite{chetrite2013} the notion of asymptotic equivalence of processes in order to construct ``modified'' Markov processes that are equivalent, in terms of typical properties, to Markov processes conditioned on reaching certain large deviations. What is found in general is that the conditioned Markov processes are not Markovian, but do become asymptotically equivalent in the long-time limit to a homogeneous Markov process, given by a generalization of the Doob transform. For more information on this large deviation conditioning problem, and its connections with nonequilibrium versions of the microcanonical and canonical ensembles, we refer to \cite{chetrite2013,chetrite2014,chetrite2015}.

To close this section, let us consider as an example an ergodic diffusion $X_t$ with path measure $P_T$, and let $Q_T$ be the path measure of the same process reversed in time (in the sense of Haussmann and Pardoux \cite{haussmann1986}). If the process is reversible, that is, if it satisfies the detailed balance condition, then it is known that the action $W_T$, which corresponds to the entropy production \cite{lebowitz1999}, depends only on the initial and final states:
\be
W_T = -\frac{1}{T}\log \frac{p(X_0)}{p(X_T)},
\ee
where $p$ is stationary density of $X_t$. Since this density does not scale with time, $X_t$ and its time-reversal must therefore be asymptotically equivalent,  and so equivalent at the level of typical values. This is expected physically, since the two processes are then statistically indistinguishable. On the other hand, if $X_t$ is irreversible, then the entropy production is known to be strictly positive, which means that the process and its time-reversal are not asymptotically equivalent. In this case, the two processes behave differently in terms of path statistics and typical ergodic values.

%%%%%%%%%%%%%%%%%%%%%%%%%%%%%%%%%%%%%%%%%%%%%%%%%%%%%%%%%%%%%%%%%%%%%%%%%%%%%%
%%%%%%%%%%%%%%%%%%%%%%%%%%%%%%%%%%%%%%%%%%%%%%%%%%%%%%%%%%%%%%%%%%%%%%%%%%%%%%
\begin{acknowledgements}
I am grateful to Frank den Hollander for comments on a first version of this paper. This work was supported by the National Research Foundation of South Africa (Grants 90322 and 96199).
\end{acknowledgements}

%%%%%%%%%%%%%%%%%%%%%%%%%%%%%%%%%%%%%%%%%%%%%%%%%%%%%%%%%%%%%%%%%%%%%%%%%%%%%%
\bibliography{masterbib}
\end{document}